\documentclass[a4paper,11pt]{fullverllncs}
\usepackage[left=2.5cm,top=2.5cm,right=2.5cm,centering]{geometry}
\usepackage{graphicx}

\usepackage{amsmath,amssymb}
\usepackage{graphicx}
\usepackage{ascmac}
\usepackage{array}
\usepackage{algpseudocode}
\usepackage{amsfonts}
\usepackage{amssymb}
\usepackage{amsmath}
\usepackage{algorithm, algpseudocode}
\usepackage{multirow}
\usepackage{arydshln}
\usepackage{hhline} 
\usepackage[misc,geometry]{ifsym} 

\usepackage[dvipdfmx, colorlinks]{hyperref, xcolor}
\definecolor{winered}{rgb}{0.5,0,0}
\definecolor{darkblue}{rgb}{0,0,0.5}
\definecolor{darkgreen}{rgb}{0,0.3,0}
\hypersetup{
linkcolor=winered,
citecolor=darkblue,
urlcolor=darkgreen
}
\newcommand{\A}{\mathsf{A}}
\newcommand{\B}{\mathsf{B}}

\newcommand{\Adv}{\mathsf{Adv}}

\newcommand{\Z}{\mathbb{Z}}
\newcommand{\G}{\mathbb{G}}
\newcommand{\N}{\mathbb{N}}

\newcommand{\Inst}{\mathtt{inst}}
\newcommand{\instsign}{\mathtt{sign}}
\newcommand{\instskip}{\mathtt{skip}}
\newcommand{\accept}{\mathtt{accept}}
\newcommand{\reject}{\mathtt{reject}}

\newcommand{\BG}{\mathsf{BG}}
\newcommand{\BGcal}{\mathcal{BG}}
\newcommand{\Ours}{\mathsf{Ours}}

\newcommand{\DS}{\mathsf{DS}}
\newcommand{\sk}{\mathsf{sk}}
\newcommand{\pk}{\mathsf{pk}}

\newcommand{\pp}{\mathsf{pp}}

\newcommand{\Setup}{\mathsf{Setup}}
\newcommand{\KGen}{\mathsf{KGen}}
\newcommand{\Sign}{\mathsf{Sign}}
\newcommand{\Verify}{\mathsf{Verify}}
\newcommand{\Agg}{\mathsf{Agg}}
\newcommand{\AVer}{\mathsf{AVer}}
\newcommand{\EUFCMA}{\mathsf{EUF \mathchar`- CMA}}

\newcommand{\rmEUFCMA}{\mathrm{EUF}\mathchar`-\mathrm{CMA}}

\newcommand{\Cert}{\mathsf{Cert}}

\newcommand{\PS}{\mathsf{PS}}
\newcommand{\GPS}{\mathsf{GPS}}

\newcommand{\AGHo}{\mathsf{AGH1}}
\newcommand{\AGHt}{\mathsf{AGH2}}
\newcommand{\LLY}{\mathsf{LLY}}
\newcommand{\AS}{\mathsf{AS}}
\newcommand{\BGLS}{\mathsf{BGLS}}

\newcommand{\SAS}{\mathsf{SAS}}

\newcommand{\SASSign}{\mathsf{SAS.Sign}}

\newcommand{\SASAVer}{\mathsf{SAS.AVer}}

\spnewtheorem{assumption}{Assumption}{\bfseries}{\itshape}

\begin{document}

\title{Pointcheval-Sanders Signature-Based \\Synchronized Aggregate Signature\thanks{A preliminary version \cite{TT22} of this paper is appeared in Information Security and Cryptology  {ICISC} 2022  - 25th International Conference.}}
\author{Masayuki Tezuka\inst{1}\textsuperscript{(\Letter)} \and Keisuke Tanaka\inst{2}}
\authorrunning{M.Tezuka et al.}

\institute{National Institute of Technology, Tsuruoka College, Yamagata, Japan\\
\email{tezuka.m@tsuruoka-nct.ac.jp} \and
Tokyo Institute of Technology, Tokyo, Japan\\
\email{keisuke@is.titech.ac.jp}
}

\maketitle
\pagestyle{plain}
\noindent
\makebox[\linewidth]{March 31, 2023}              

\begin{abstract}
Synchronized aggregate signature is a special type of signature that all signers have a synchronized time period and allows aggregating signatures which are generated in the same period.
This signature has a wide range of applications for systems that have a natural reporting period such as log and sensor data, or blockchain protocol.

In CT-RSA 2016, Pointcheval and Sanders proposed the new randomizable signature scheme.
Since this signature scheme is based on type-3 pairing, this signature achieves a short signature size and efficient signature verification.

In this paper, we design the Pointchcval-Sanders signature-based synchronized aggregate signature scheme and prove its security under the generalized Pointcheval-Sanders assumption in the random oracle model.
Our scheme offers the most efficient aggregate signature verification among synchronized aggregate signature schemes based on bilinear groups.

\end{abstract}

\section{Introduction}
\subsection{Background}\label{IntroBack}
\paragraph{\bf Aggregate Signature.}

Aggregate signature originally introduced by Boneh, Gentry, Lynn, and Shacham \cite{BGLS03} allows anyone to compress many signatures produced by different signers on different messages into a short aggregate signature.
The size of an aggregate signature size is the same as any signature.
By verifying an aggregate signature, we can check the validity of all those individual signatures which are compressed into an aggregate signature.

These attractive features are useful for the internet of things (IoT) system to reduce the storage space for signatures and realize efficient verification of signatures.
An aggregate signature scheme is expected to be used in a wide range of applications such as Border Gateway Protocol (BGP) routing \cite{BGOY07}, certificate chain compression \cite{BGLS03}, bundling software updates \cite{AGH10}, sensor network data \cite{AGH10}, or blockchain protocol~\cite{HW18}.

Currently, only three aggregate signature scheme constructions are known.
The first construction by Boneh et al. \cite{BGLS03} is based on bilinear maps.
This scheme can aggregate signatures as well as already aggregated signatures (i.e., full aggregation) in any order.
The security of this scheme is proven under the co-computational Diffie-Hellman (co-CDH) assumption in the random oracle model (ROM).
However, their scheme has a drawback in that the verification cost of an aggregate signature is expensive.
Concretely, the number of pairing operations in verification for an aggregate signature is proportional to the number of signatures compressed into the aggregate signature.

The other schemes are constructed in the standard model (without the ROM).
The second scheme by Hohenberger, Sahai, and Waters \cite{HSW13} is based on multilinear maps.
The third scheme by Hohenberger, Koppula, and Waters \cite{HKW15} is an indistinguishability obfuscation (iO) based construction.
Since constructing aggregate signature schemes from standard computational assumptions without the ROM is a difficult task, several variants of aggregate signature with restricted aggregation have been proposed.

\paragraph{\bf Synchronized Aggregate Signature.}
One variant of aggregate signature is synchronized aggregate signature.
The concept of this signature was proposed by Gentry and Ramzan \cite{GR06}.
They constructed an identity-based aggregate signature that is based on the computational Diffie-Hellman (CDH) assumption in the ROM.

After their seminal work, Ahn, Green, and Hohenberger \cite{AGH10} revisited their model and proposed a synchronized aggregate signature.
In this scheme, all of the signers have a synchronized time period.
For each time period, each signer can sign a message at most once and signatures generated in the same time period only can be compressed into an aggregate signature.
Even though a synchronized aggregate signature scheme has restrictions described above, it is still useful for systems that have a natural reporting period.
(e.g. log data \cite{AGH10}, sensor data \cite{AGH10}, blockchain protocols \cite{HW18})

So far, several synchronized aggregate signature schemes were proposed.
Ahn, Green, and Hohenberger \cite{AGH10} gave a pairing-based based synchronized aggregate signature scheme based on the CDH assumption without the ROM.
Moreover, they also gave an efficient pairing-based based synchronized aggregate signature scheme whose security is proven under the CDH assumption in the ROM.

Lee, Lee, and Yung \cite{LLY13} gave a synchronized aggregate signature scheme based on the Camenisch-Lysyanskaya (CL) signature scheme \cite{CL04}.
The security of this scheme relies on an interactive assumption called Lysyanskaya-Rivest-Sahai-Wolf (LRSW) assumption \cite{LRSW99} in the ROM.
Tezuka and Tanaka \cite{TT20} revisited their security analysis result and improved it by showing the security based on a non-interactive assumption called the modified 1-strong Diffie-Hellman-2 (1-MSDH-2) assumption \cite{PS18} in the ROM.

As for a pairing-free scheme, Hohenberger and Waters \cite{HW18} proposed the synchronized aggregate signature scheme based on the RSA assumption without the ROM.

\paragraph{\bf Motivation: Efficient Synchronized Aggregate Signature.}
In pairing-based synchronized aggregate signature schemes, the scheme by Lee et al. \cite{LLY13} is the most efficient synchronized aggregate signature scheme.
Their scheme offers the smallest number of pairing operations (3 pairing operations) in an aggregate signature verification (See Fig.\ref{SASListed}). 
From the viewpoint of the efficiency of aggregate signature verification, it is desirable to construct a synchronized aggregate signature scheme with fewer pairing operations for aggregate signature verification.

\subsection{Our Result}
\paragraph{\bf Our Result.}
In this paper, we give a new synchronized aggregate signature scheme based on the Pointcheval-Sanders (PS) signature scheme \cite{PS16}.
The security of our scheme can be proven under the generalized Pointcheval-Sanders (GPS) assumption  \cite{KLAP21} in the ROM.

In general, compared to the computation cost of multiplication for elliptic curve points, the computation of pairing is more costly.
To clarify the advantages of our synchronized aggregate scheme, we compare our scheme with other schemes (See Fig.\ref{SASListed}).

\begin{figure}[h]
\begin{center}
\begin{tabular}{llccccccc}\hline

Scheme &Assumption & $pp$ & ~$\pk$~ & Agg & Agg Ver & Pairing & CertKey\\
 &   & size & size &size & (Pairing op) & type &model \\
\hline
\hline
$\AS_{\BGLS}{}^{\dagger}$ \cite{BGLS03} \S 3 &co-CDH + ROM  &$O(1)$ &1 &2 &$n+1$ & Type-2  & \checkmark\\
$\SAS_{\AGHo}$ \cite{AGH10} \S 4 &CDH  &$O(k)$ &1 &3 &$k+3$ & Type-3  &\checkmark  \\
$\SAS_{\AGHt}$ \cite{AGH10} \S A \ &CDH + ROM &$O(1)$ &1 &3 &4 & Type-3 &\checkmark \\
$\SAS_{\LLY}$ \cite{LLY13}&1-MSDH-2 + ROM \ &$O(1)$ &1 &2 &3 &  Type-1 &\checkmark \\
$\SAS_{\Ours}$ \S \ref{Secourcon}  &GPS + ROM &$O(1)$ &2 &2 &2 &  Type-3 &\checkmark \\
\hline
\end{tabular}\\
\end{center}

\caption{\small Comparison with pairing-based synchronized aggregate signature schemes.
In the column of $``$Assumption$"$, $``$ROM$"$ represents the random oracle model.
In the columns of $``$pp size$"$, $``\pk$ size$"$, $``$Agg size$"$ represent the number of elements in a public parameter pp, a public key $\pk$, and an aggregate signature, respectively.  
In the column of $``$Agg Ver (Pairing op)$"$ represents the number of pairing operations in the verification of an aggregate signature.
In the column of $``$CertKey model $"$, $``\checkmark"$ represents that the EUF-CMA security of the corresponding scheme is proven in the certified-key model.\\
In $\AS_{\BGLS}$, $n$ represents the number of original signatures which are aggregated into an aggregate signature.
$\AS_{\BGLS}$ can be used as a synchronized aggregate signature scheme, with the following trivial modifications.
A message $m$ is changed to a message-period pair $(m, t)$. Aggregation of signature is only allowed for signatures that are signed in the same time period $t$.
An aggregate signature of $\AS_{\BGLS}$ is composed $1$ element, but in other synchronized aggregate signature schemes, information of time period $t$ is included in an aggregate signature. 
For fair comparison to other synchronized aggregate signature schemes,  we include $t$ into an aggregate signature and count the number of elements in an aggregate signature as $2$.
Security of $\AS_{\BGLS}$ simply can be proven under the co-CDH assumption in the ROM under the aggregation restriction that signatures for the same message cannot be aggregated.
Without this aggregation restriction, $\AS_{\BGLS}$ can be used as a multi-signature, however, it falls victim to the rogue key attack which is known as a notorious attack for multi-signature schemes \cite{BDN18}.
In synchronized aggregate signature has a restriction that each signer issues a signature one-time for each period, but it allows aggregating signatures on the same message.
To prevent the rogue key attack, we should pose the certified-key model for $\AS_{\BGLS}$.
In  $\SAS_{\AGHo}$ has a $\ell \times k$-bits message space ($k$ chunks of $\ell$-bits message). \\
}

\label{SASListed}
\end{figure}

\paragraph{\bf Comparison with Other Schemes.}
The scheme $\BGLS$ \cite{BGLS03} is a full-aggregate signature scheme that offers optimal public-key size and aggregate signature size.
A full-aggregate signature scheme can be used as a synchronized aggregate signature scheme, with the following trivial modifications.
A message $m$ is changed to a message-period pair $(m, t)$.
Aggregation of signatures is only allowed for signatures that are signed in the same time period $t$.
However, if we use $\BGLS$ as a synchronized aggregate signature scheme, $n+1$ pairing operations are needed for verifying an aggregate signature where $n$ is the number of aggregated original signatures. 

The scheme $\SAS_{\AGHo}$ \cite{AGH10} is a synchronized aggregate signature scheme in the standard model.
In $\SAS_{\AGHo}$, message space is $\ell \times k$-bits message space. ($k$ chunks of $\ell$-bits strings).
If we set $k=1$ in $\SAS_{\AGHo}$, $k+3 = 4$ pairing operations is needed for verifying an aggregate signature.

The  $\SAS_{\AGHt}$ \cite{AGH10}  and  $\SAS_{\LLY}$ \cite{LLY13} are synchronized aggregate signature schemes in the random oracle model.
In both schemes, a public key is composed of $1$ group element.
 $\SAS_{\AGHt}$ needs $4$ paring operations and $\SAS_{\LLY}$  needs $3$ paring operations for verifying an aggregate signature, respectively.
Although a public key of our scheme is composed of $2$ group elements, our scheme only needs $2$ paring operations for verifying an aggregate signature.

Thus, compared with existing paring-based synchronized aggregate signature schemes, our scheme offers the fewest paring operations in a verification of an aggregate signature.
Our scheme offers the most efficient aggregate signature verification among synchronized aggregate signature schemes based on bilinear groups.

\subsection{Technical Overview}\label{TechOverview}
\paragraph{\bf How to Construct Our Signature Scheme.}
The core idea of our construction is based on the combination of randomizable signature, the ``public-key sharing technique" and the ``randomness re-use technique" \cite{LOSSW06}.
These technique are used to construct variants of aggregate signatures scheme \cite{LOSSW06,Sch11,LLY13,CK20}. 

Lee et al \cite{LLY13} used these techniques to construct a synchronized aggregate signatures scheme based on the CL signature scheme which is a randomizable signature scheme.
The security of these schemes can be proven by the security of the original (CL) signature scheme.

\paragraph{\bf Problem in Security Proof.}
However, it is not clear that it is possible to design a PS signature-based synchronized aggregate signature scheme with provable security.
Since existing CL signature-based synchronized aggregate signature scheme  $\SAS_{\LLY}$ \cite{LLY13}  is given in only type-$1$ pairing, a type-$3$ pairing variant of CL signature-based synchronized aggregate signature scheme is not known.

Our first attempt is to apply the public-key sharing technique and the randomness re-use technique to the PS signature scheme which is also a randomizable signature scheme.
In fact, we obtain the PS signature-based synchronized signature scheme but we fail to prove our scheme from the EUF-CMA security of the original (PS) signature scheme.

Now, we briefly explain the reason why the security proof technique in \cite{LLY13} fails in our scheme.
In $\SAS_{\LLY}$, a group element of a public-key and group elements of signature belong to the same group $\G$.
This fact allows signature simulation in the security proof of $\SAS_{\LLY}$ scheme.
In the security proof of $\SAS_{\LLY}$, by using the programmability of the random oracle model, a signature is generated by computing multiplications of public-key.

By contrast, in our construction, group elements of signature and a group element of signature belong to different groups (See Fig. \ref{PSASsigconst}).
Group elements of a public-key $(\tilde{X}, \tilde{Y})$ belong to the group  $\widetilde{\G}$ and a group element of signature $B$ belongs to the group $\G$.
If we try to generate a signature by multiplying public-key elements $\tilde{X}$ and $\tilde{Y}$, the result of the multiplication does not belong to $\G$.
Thus, the security proof technique by \cite{LLY13} cannot be applied to our scheme.

\paragraph{\bf Our Approach for Security Proof.}
To prove the security of our scheme, we use the generalized PS  (GPS) assumption \cite{KLAP21} which is a variant of the PS assumption \cite{PS16}.
These assumptions are classified into interactive assumptions.
The interactive assumption is that the computational problem is difficult for all probabilistic polynomial time adversary which tries to solve the problem even if oracle queries that are related to the problem are allowed.

Briefly, the difference between the PS assumption and the GPS assumption is equipped oracles (See Assumption \ref{PSassum} and Assumption \ref{GPSassum}).
The GPS assumption is obtained by changing the oracle equipped with the PS assumption as follows.
We divide the computation of the equipped oracle in the PS assumption into $2$ computation steps and replace the equipped oracle with 2 oracles that compute each step.
By using $2$ oracles in the GPS assumption, we prove the security of our scheme under the GPS assumption in the random oracle model.

\subsection{Related Works}

\paragraph{\bf Variants of Aggregate Signature.}
An aggregate signature can be categorized into various types from the point of view of aggregation restriction.
The full aggregate signature proposed by Boneh et. al \cite{BGLS03} allows any user to aggregate signatures generated by different signers.
Moreover, this scheme allows us to aggregate individual signatures as well as already aggregated signatures in any order.

Lysyanskaya, Micali, Reyzin, and Shacham \cite{LMRS04} proposed sequential aggregate signature.
This signature scheme allows a signer to add his signature to an aggregate signature in sequential order.

Synchronized aggregate signature scheme \cite{GR06,AGH10} allows signers to generate at most one signature for each period and aggregate signatures generated in the same period into an aggregate signature.

Chalkias, Garillot, Kondi, and Nikolaenko \cite{CGKN21} proposed the notion of half-aggregation.
Half-aggregation allows compressing signatures into an aggregate signature that has half size of the total signature size.

Hartung, Kaidel, Koch, Koch, and Rup \cite{HKKKR16} proposed fault-tolerant aggregate signature.
In this signature, as long as the number of invalid signatures aggregated does not exceed a certain bound, a verification algorithm can determine a subset of all messages belonging to an aggregate that were signed correctly. 

Goyal and Vaikuntanathan \cite{GV22} proposed locally verifiable aggregate signature. 
In this scheme, given an aggregate signature corresponding to the set of $M$ of $n$ messages, a local verification algorithm can check whether a particular message $m$ is in the set $M$.
Moreover, the runtime of a local verification algorithm is independent of $N$ and the local verification algorithm can be run without knowledge of the entire set $M$.

\paragraph{\bf Pointcheval-Sanders Signature.}
The Pointcheval-Sanders (PS) signature scheme \cite{PS16} is a randomizable signature scheme that allows anyone to refresh a valid signature $\sigma$ on a message $m$ to a new valid signature $\sigma'$ on the same message $m$.
Compared to the Camenisch-Lysyanskaya signature scheme \cite{CL04} which is also a randomizable signature scheme, this scheme offers a short signature size.

Security of this signature scheme was proven under the interactive assumption called the PS assumption \cite{PS16}.
In \cite{PS18}, Pointcheval and Sanders introduced the non-interactive assumption called the modified $q$-strong Diffie-Hellman-1 ($q$-MSDH-1) assumption.
They proved the weak-EUF-CMA security of the PS signature scheme from the $q$-MSDH-1 assumption.

The PS signature scheme (the PS assumption) and its variant are important starting points to construct signature schemes with functionalities.
(e.g. sequential aggregate signature \cite{PS16,McD20}, redactable signature \cite{McD20,San20},  threshold signature \cite{ADEO21}, group signature \cite{CS20,KLAP21,KSAP21,San21,ST21}, threshold group signature \cite{CDLNT20}, multi-signature \cite{CDLNT20}, updatable signature \cite{CRSST21})
Moreover, relationships between the PS signature and the structure-preserving signature have been studied.

Gardafi \cite{Gha21} introduced the notion of a partially structure-preserving signature. 
In a structure-preserving signature scheme \cite{AFGHO10}, all the messages, signatures, and public keys are group elements.
Partially-preserving signature is the same with the exception that the message space is $\mathbb{Z}^n_{p}$ where $n$ is an integer and $p$ is a prime.
They further proposed the notion of linear-massage strongly partially structure-preserving signature where the message is embedded in a linear manner.
This signature class includes the CL signature scheme and the PS signature scheme.
They proved some impossibility results and lower bound results for a linear-massage strongly partially structure-preserving signature and gave a generic transformation from a linear-massage strongly partially structure-preserving signature scheme to a structure-preserving signature scheme.

In recent work by Sedaghat, Slamanig, Kohlweiss, and Preneel \cite{SSKP22}, they introduced the notion of  a message-indexed structure-preserving signature which is a variant of a structure-preserving signature whose message is parameterized by a message indexing function. 
They gave a message-indexed structure-preserving signature scheme whose construction is inspired by the PS signature scheme and the structure-signature scheme by Ghadafi \cite{Gha16}.
Moreover, they proposed a notion of a structure-preserving threshold signature and gave a construction based on a message-indexed structure-preserving signature scheme.

\subsection{Road Map}
In Section \ref{Prelimi}, we recall pairing groups and a digital signature.
In Section \ref{SecSyncAggSig}, we review synchronized aggregate signature scheme and its security.
In Section \ref{SecSyncSASconst}, we review the PS signature scheme, provide a high-level idea of our construction, and give our synchronized aggregate signature and prove its security.

\section{Preliminaries}\label{Prelimi}
In this section, we introduce notations and review pairing groups and the Pointcheval Sanders assumption.
Then, we review a digital signature scheme.

\paragraph{\bf Notations.}
Let $1^{\lambda}$ be the security parameter. 
A function $f$ is negligible in $k$ if $f(k) \leq 2^{-\omega(\log k)}$.
For a positive integer $n$, we define $[n]: =\{1,\dots, n\}$.
For a finite set $S$, $s \xleftarrow{\$} S$ represents that an element $s$ is chosen from $S$ uniformly at random. 
For a group $\G$, we define $\G^* := \G \backslash \{1_{\G}\}$. 
For an algorithm $\A$, $y \leftarrow \A(x)$ denotes that the algorithm $\A$ outputs $y$ on input~$x$.
We abbreviate probabilistic polynomial time as PPT.

\subsection{Bilinear Group}
A pairing group is a tuple $\BGcal = (p, \G, \widetilde{\G}, \G_T, e)$ where $\G$, $\widetilde{\G}$ and $\G_T$ are cyclic group of prime order $p$ and $e:\G \times \widetilde{\G} \rightarrow \G_T$ is an efficient computable, non-degenerating bilinear map. (i.e., $e$ satisfies the following properties.)
\begin{enumerate}
\item For all $X \in \G$, $\widetilde{Y} \in \widetilde{\G}$ and $a, b \in \Z_{p}$, then $e(X^a, \widetilde{Y}^b) = e(X, \widetilde{Y})^{ab}$.
\item For all $G \in \G^*$, $\widetilde{G} \in \widetilde{\G}^*$, $e(G, \widetilde{G}) \neq 1_{\G_T}$.
\end{enumerate}
Type-3 pairing groups \cite{GPS08} are pairing groups which satisfy $\G \neq \widetilde{\G}$ and there is no efficiently computable homomophism from $\widetilde{\G}$ to $\G$.

We introduce a type-3 bilinear group generator.
A type-3 bilinear group generator $\BG$ is an algorithm that takes as an input a security parameter $1^{\lambda}$.
Then, it returns the descriptions of an asymmetric pairing $\BGcal = (p, \G, \widetilde{\G}, \G_T, e)$ where $p$ is a $\lambda$-bits prime.

Pointcheval and Sanders \cite{PS16} introduced the interactive assumption called Pointcheval-Sanders (PS) assumption.
This assumption holds in the generic group model~\cite{Sho97}.

\begin{assumption}[PS Assumption \cite{PS16}]\label{PSassum}
Let $\BG$  be a type-3 bilinear group generator and $\A$ be a PPT algorithm.
The Pointcheval-Sanders (PS) assumption over $\BG$ is defined by the game $\PS_{\BG}$ in Fig.\ref{PSgame}.

\begin{figure}[h]
\centering
\begin{tabular}{|l|}
\hline
GAME $\PS^{\A}_{\BG}(\lambda):$\\
~~~$Q \leftarrow \{\}$, $\BGcal=(p, \G, \widetilde{\G}, \G_T, e) \leftarrow \BG(1^{\lambda})$, $G \xleftarrow{\$} \G^*$, $\widetilde{G} \xleftarrow{\$}  \widetilde{\G}^*$,\\
~~~$x, y \xleftarrow{\$} \Z^*_p$, $\widetilde{X} \leftarrow \widetilde{G}^{x}$, $\widetilde{Y} \leftarrow \widetilde{G}^{y}$, $(A^*, B^*, m^* ) \leftarrow \A^{\mathcal{O}_{x,y}(\cdot)}(\BGcal, G^*, \widetilde{G}^*, \widetilde{X}, \widetilde{Y})$\\
~~~If $m^* \notin Q \land A^* \neq 1_{\G} \land B^*= (A^*)^{x+m^*\cdot y} $, return $1$. Otherwise, return $0$
\\
\\
$\mathcal{O}_{x, y} (m):$\\
~~~$Q \leftarrow Q \cup \{m\}$, $A \xleftarrow{\$} \G^*$, return $(A, A^{x + m \cdot y})$\\
\hline
\end{tabular}
\caption{\small
The game $\PS^{\A}_{\BG}$.}
\label{PSgame}
\end{figure}

The advantage of an adversary $\A$ in the game $\PS_{\BG}$ is defined by $\Adv^{\PS}_{\BG, \A}(\lambda) \allowbreak := \Pr[1 \Leftarrow \PS^{\A}_{\BG}(\lambda)]$.
We say that the PS assumption holds if $\Adv^{\PS}_{\BG, \A}(\lambda)$ is negligible in $\lambda$ for all PPT adversaries $\A$.  
\end{assumption}

Kim, Lee, Abdalla, and Park proposed the generalized Pointcheval-Sanders (GPS) assumption \cite{KLAP21}.
This assumption is a modification of the PS assumption in that  the oracle $\mathcal{O}_{x,y}(\cdot)$ in the PS assumption is divided into the following two oracles.
$\mathcal{O}^{\GPS}_{0}$ samples a group element $A$ and $\mathcal{O}^{\GPS}_{1}$ computes $B \leftarrow A^{x+m\cdot y}$ where $(A, m)$ is given to $\mathcal{O}^{\GPS}_{1}$ as an input.

\begin{assumption}[GPS Assumption \cite{KLAP21}]\label{GPSassum}
Let $\BG$  be a type-3 bilinear group generator and $\A$ be a PPT algorithm.
The generalized Pointcheval-Sanders (GPS) assumption over $\BG$ is defined by the game $\GPS_{\BG}$ in Fig.\ref{GPSgame}.

\begin{figure}[h]
\centering
\begin{tabular}{|l|}
\hline
GAME $\GPS^{\A}_{\BG}(\lambda):$\\
~~~$Q_{0}, Q_{1} \leftarrow \{\}$, $\BGcal=(p, \G, \widetilde{\G}, \G_T, e) \leftarrow \BG(1^{\lambda})$, $G \xleftarrow{\$} \G^*$, $\widetilde{G} \xleftarrow{\$}  \widetilde{\G}^*$,\\
~~~$x, y \xleftarrow{\$} \Z^*_p$, $\widetilde{X} \leftarrow \widetilde{G}^{x}$, $\widetilde{Y} \leftarrow \widetilde{G}^{y}$, $(A^*, B^*, m^* ) \leftarrow \A^{\mathcal{O}^{\GPS}_{0} (), \mathcal{O}^{\GPS}_{1} (\cdot, \cdot)}(\BGcal, G, \widetilde{G}, \widetilde{X}, \widetilde{Y})$\\
~~~If $(\cdot, m^*) \notin Q_{1} \land A^* \neq 1_{\G} \land B^*= (A^*)^{x+m^*\cdot y} $, return $1$. Otherwise, return $0$
\\
\\
~~~$\mathcal{O}^{\GPS}_{0} ():$\\
~~~~~~$A \xleftarrow{\$}  \G^*$, $Q_{0} \leftarrow Q_{0} \cup \{A\}$, return $A$\\
~~~$\mathcal{O}^{\GPS}_{1} (A, m \in \mathbb{Z}_{p}):$\\
~~~~~~If $(A \notin Q_{0} \lor (A, \cdot)  \in Q_{1})$, return $\bot$.\\
~~~~~~$B \leftarrow A^{x + m \cdot y}$, $Q_{1} \leftarrow Q_{1} \cup \{(A, m)\}$, return $B$.\\
\hline
\end{tabular}
\caption{\small
The game $\GPS^{\A}_{\BG}$.}
\label{GPSgame}
\end{figure}

The advantage of an adversary $\A$ in the game $\GPS_{\BG}$ is defined by $\Adv^{\GPS}_{\BG, \A}(\lambda) \allowbreak := \Pr[1 \Leftarrow \GPS^{\A}_{\BG}(\lambda)]$.
We say that the GPS assumption holds if $\Adv^{\GPS}_{\BG, \A}(\lambda)$ is negligible in $\lambda$ for all PPT adversaries $\A$.  
\end{assumption}

Kim et al. \cite{KLAP21} proved that the GPS assumption holds in the generic group model.
Moreover, Kim, Sanders, Abdalla, and Park \cite{KSAP21} analyzed the relationship among the PS assumption, the GPS assumption, and the symmetric discrete logarithm assumption.
More precisely, from their result, the following facts are clarified.

\begin{itemize}
\item If the GPS assumption holds, the PS assumption holds.
\item If the symmetric discrete logarithm assumption holds, the GPS assumption holds in the algebraic group model \cite{FKL18}.
\end{itemize}

\subsection{Digital Signature Scheme}
We review a digital signature scheme and its security notion.
\begin{definition}[Digital Signature Scheme]
A digital signature scheme $\DS$ consists of following four algorithms $(\Setup, \KGen, \allowbreak \Sign, \Verify)$.
\begin{itemize}
\item $\Setup (1^{\lambda}):$ A setup algorithm takes as an input a security parameter $1^{\lambda}$. It returns the public parameter $\pp$.
In this work, we assume that $\pp$ defines a message space and represents this space by  $\mathcal{M}_{\pp}$.
We omit a public parameter $\pp$ in the input of all algorithms except for $\KGen$.

\item $\KGen (\pp):$ A key-generation algorithm takes as an input a public parameter $\pp$. It returns a public key $\pk$ and a secret key $\sk$.

\item $\Sign (\sk, m):$ A signing algorithm takes as an input a secret key $\sk$ and a message $m$. It returns a signature $\sigma$.

\item $\Verify (\pk, m, \sigma):$ A verification algorithm takes as an input a public key $\pk$, a message $m$, and a signature $\sigma$.
It returns a bit $b \in  \{0, 1\}$.
\end{itemize}
\paragraph{\bf Correctness.}
$\DS$ satisfies correctness if for all $\lambda \in \N$, $\pp \leftarrow \Setup (1^{\lambda})$ for all $m \in \mathcal{M}_{\pp}$, $(\pk, \sk) \leftarrow \KGen(\pp)$, and $\sigma \leftarrow \Sign(\sk, m)$, $\Verify(\pk, m, \sigma) = 1$ holds.
\end{definition}

We review a security notion called the existentially unforgeable under chosen message attacks $(\rmEUFCMA)$ security for digital signature.

\begin{definition}[EUF-CMA Security]
The existentially unforgeable under chosen message attacks $(\rmEUFCMA)$ security of a digital signature scheme $\DS$ is defined as Fig. \ref{EUFCMAgame}.

\begin{figure}[h]
\centering
\begin{tabular}{|l|}
\hline
GAME $\EUFCMA^{\DS}_{\A}:$\\
~~~$Q \leftarrow \{\}$, $\pp \leftarrow \Setup (1^{\lambda})$, $(\pk, \sk) \leftarrow \KGen(\pp)$, $(m^*, \sigma^*) \leftarrow \A^{\mathcal{O}^{\Sign}(\cdot)}(\pp, \pk)$\\
~~~If $\Verify(\pk, m^*, \sigma^*) = 1 \land ~ m^* \notin Q$, return $1$. Otherwise return $0$.\\
\\
Oracle $\mathcal{O}^{\Sign}(m):$\\
~~~$Q \leftarrow Q \cup \{m\}$, $\sigma \leftarrow \Sign(\sk, m)$, return $\sigma$.\\
\hline
\end{tabular}
\caption{\small
The $\rmEUFCMA$ security game $\EUFCMA^{\DS}_{\A}$.}
\label{EUFCMAgame}
\end{figure}

The advantage of an adversary $\A$ for the $\rmEUFCMA$ security game is defined by $\Adv^{\EUFCMA}_{\DS, \A}:= \Pr[\EUFCMA^{\DS}_{\A} \Rightarrow 1]$.
$\DS$ satisfies $\rmEUFCMA$ security if for all PPT adversaries $\A$, $\Adv^{\EUFCMA}_{\DS, \A}$ is negligible in $\lambda$.
\end{definition}

\section{Synchronized Aggregate Signature} \label{SecSyncAggSig}
In this section, we review a synchronized aggregate signature scheme and it security model.

\subsection{Synchronized Aggregate Signature Scheme}
An aggregate signature \cite{BGLS03} allows us to compress an arbitrary number of individual signatures into a short aggregate signature.
A synchronized aggregate signature \cite{AGH10} is a variant of aggregate signature that all signers have a synchronized time clock or has an access to the public current time period.
For each time period $t$, each signer can sign a message at most once and anyone can aggregate signatures generated by different signers in the same period $t$.
A generated aggregate signature is the same size as an individual signature.

Now, we review a definiton of a synchronized aggregate signature.

\begin{definition} [Synchronized Aggregate Signature Scheme \cite{AGH10,GR06}]\label{SyncAggDef}
A synchronized aggregate signature scheme $\SAS$ for a bounded number of periods is a tuple of algorithms $(\Setup, \KGen,\allowbreak \Sign, \Verify,\allowbreak \Agg, \allowbreak \AVer)$.

\begin{itemize}
\item $\Setup(1^\lambda, 1^T): $ A setup algorithm takes as an input a security parameter $\lambda$ and the time period bound $T$.
It returns the public parameter $\pp$.
We assume that $\pp$ defines the message space $\mathcal{M}_{\pp}$.
We omit a public parameter $\pp$ in the input of all algorithms except for $\KGen$.

\item $\KGen (\pp):$ A key-generation algorithm takes as an input a public parameter $\pp$.
It returns a public key $\pk$ and a secret key $\sk$.

\item $\Sign (\sk, t, m):$ A signing algorithm takes as an input a secret key $\sk$, a time period $t \leq T$, and a message $m$.
It returns a signature $\sigma$.
We assume that the information of time period $t$ is contained in a signature $\sigma$.

\item $\Verify (\pk, m, \sigma):$ A verification algorithm takes as an input a public key $\pk$, a message $m$, and a signature $\sigma$.
It returns a bit $b \in \{0, 1\}$.

\item $\Agg ((\pk_i, m_i, \sigma_i)_{i \in [\ell]}):$
An aggregation algorithm takes as an input a list of tuple $(\pk_i, m_i, \sigma_i)_{i \in [\ell]}$.
It return either an aggregate signature $\Sigma$ or $\bot$.  
We assume that the information of time period $t$ is contained in an aggregate signature $\Sigma$.

\item $\AVer((\pk_i, m_i)_{i \in [\ell]}, \Sigma):$
An aggregate signature verification algorithm takes as an input a list of tuple $(\pk_i, m_i)_{i \in [\ell]}$ and an aggregate signature $\Sigma$.
It returns a bit $b \in \{0, 1\}$.
\end{itemize}
\paragraph{\bf Correctness.}
$\SAS$ satisfies correctness if for all $\lambda \in \mathbb{N}$, $T \in \mathbb{N}$, $\pp \leftarrow \Setup(1^\lambda, 1^{T})$, for any finite sequence of key pairs $(\pk_1, \sk_1),\dots (\pk_\ell, \sk_\ell) \leftarrow \KGen (\pp)$ where $\pk_i$ are all distinct, for any time period $t \leq T$, for any sequence of messages $(m_1, \dots m_\ell) \in \mathcal{M}_{pp}^{\ell}$, 
$\sigma_i \leftarrow \Sign (\sk_i, t, m_i)$ for $i \in [\ell]$, 
$\Sigma \leftarrow \Agg((\pk_i, m_i, \sigma_i)_{i \in [\ell]})$, we have
\begin{equation*}
\begin{split}
\Verify &(\pk_i, m_i, \sigma_i) = 1 {\rm \ for \ all \ } i \in [\ell] \land \AVer ((\pk_i, m_i)_{i \in [\ell]}, \Sigma) = 1.
\end{split}
\end{equation*}
\end{definition}

\subsection{Security for Synchronized Aggregate Signature}
We review a security model called the existentially unforgeable under chosen message attacks (EUF-CMA) security in the certified-key model.

Gentry and Ramzan \cite{GR06} introduced the existentially unforgeable under chosen message attacks ($\rmEUFCMA$) security for synchronized aggregate signature.
In this security model, a public parameter $\pp$ and a challenge public key $\pk^*$ are given to an adversary which tries to forge an aggregate signature without secret key $\sk^*$.
For each period $t$, the adversary allows to access signing oracle $\mathcal{O}^{\Sign}$ and obtain a signature for an arbitrary message.
This security guarantees that it is hard for an adversary to forge an aggregate signature that is valid and non-trivial.
Gentry and Ramzan \cite{GR06} constructed an identity-based synchronized aggregate signature scheme.

Ahn, Green, and Hohenberger \cite{AGH10} introduced the certified-key model for a synchronized aggregate signature.
In this model, signers must prove that a tuple of keys $(\pk, \sk)$ is generated honestly by an algorithm $\KGen$. 
To prove the honest generation of a public key $\pk$, the signer (adversaries for $\rmEUFCMA$) must submit a tuple $(\pk, \sk)$ to the certification oracle $\mathcal{O}^{\Cert}$.
Now, we review the $\rmEUFCMA$ security in the certified-key model.

\begin{definition}[EUF-CMA Security in the Certified-Key Model \cite{AGH10,LLY13}]
The existentially unforgeable under chosen message attacks $(\rmEUFCMA)$ security of a synchronized aggregate signature scheme $\SAS$ in the certified-key model is defined as Fig. \ref{SASEUFCMAgame}.

\begin{figure}[h]
\centering
\begin{tabular}{|l|}
\hline
GAME $\EUFCMA^{\SAS}_{\A}:$\\
~~~$Q \leftarrow \{\}$, $L \leftarrow  \{\}$, $t \leftarrow 1$, $\pp \leftarrow \Setup (1^{\lambda})$, $(\pk^*, \sk^*) \leftarrow \KGen(\pp)$, \\
~~~$((\pk^*_i, m^*_i)_{i \in [\ell^*]}, \Sigma^*) \leftarrow \A^{\mathcal{O}^{\Cert}(\cdot, \cdot), \mathcal{O}^{\Sign}(\cdot, \cdot)}(\pp, \pk^*)$\\
~~~If $(\SASAVer((\pk^*_i, m^*_i)_{i \in [\ell^*]}, \Sigma^*) = 1)$\\
~~~~~~$\land$ (for all  $i \in [\ell^*]$  such  that $\pk^*_j \neq \pk^*$, $\pk^*_j \in L$)\\
~~~~~~$\land$ ($\pk^*_{j^*} = \pk^* \land m^*_{j^*} \notin Q$ for some $j^* \in [\ell^*]$), return $1$.\\
~~~Otherwise return $0$.\\
\\
Oracle $\mathcal{O}^{\Cert}(\pk, \sk):$\\
~~~If the key pair $(\pk, \sk)$ is valid, $L \leftarrow L \cup \{\pk\}$ and return $``\accept"$.\\
~~~Otherwise, return $``\reject"$.\\
Oracle $\mathcal{O}^{\Sign}(\Inst, m):$\\
~~~If $\Inst = \instskip$, $t \leftarrow t+1$.\\
~~~Otherwise, $Q \leftarrow Q \cup \{m\}$, $\sigma \leftarrow\SASSign (\sk^*, t, m)$, $t \leftarrow t+1$, return $\sigma$.\\

\hline
\end{tabular}
\caption{\small
The $\rmEUFCMA$ security game in the certified-key model $\EUFCMA^{\SAS}_{\A}$.}
\label{SASEUFCMAgame}
\end{figure}

The advantage of an adversary $\A$ for the $\rmEUFCMA$ security game in the certified-key model is defined by $\Adv^{\EUFCMA}_{\SAS, \A}:= \Pr[\EUFCMA^{\SAS}_{\A} \Rightarrow 1]$.
$\SAS$ satisfies $\rmEUFCMA$ security in the certified-key model if for all PPT adversaries $\A$, $\Adv^{\EUFCMA}_{\SAS, \A}$ is negligible in $\lambda$.

\end{definition}

\section{PS Signature-Based Synchronized Aggregate Signature} \label{SecSyncSASconst}
In this section, we review the Pointcheval-Sanders (PS) signature scheme \cite{PS16}.
Then, we give a high-level idea of our synchronized aggregate signature scheme from the PS signature scheme and give our synchronized aggregate signature scheme.
Finally, we prove the security of our scheme from the EUF-CMA security of the PS signature scheme in the ROM.

\subsection{Pointcheval-Sanders Signature Scheme \cite{PS16}}\label{PSsigconst}
Pointcheval and Sanders \cite{PS16} proposed a short randomizable signature scheme.
We review the single-message Pointcheval-Sanders (PS) signature scheme $\DS_{\PS} = (\Setup_{\PS}, \KGen_{\PS}, \Sign_{\PS}, \Verify_{\PS})$.
The construction of their scheme is described in Fig.\ref{PSsigconst}.

\begin{figure}[h]
\centering
\begin{tabular}{|l|}
\hline
$\Setup_{\PS}(1^\lambda):$\\
~~~$\BGcal= (p, \G, \widetilde{\G}, \G_T, e) \leftarrow \BG(1^\lambda)$, return $\pp \leftarrow \BGcal$.\\
 $\KGen_{\PS}(\pp):$\\
~~~$\widetilde{G} \xleftarrow{\$} \widetilde{\G}^*$, $x, y \xleftarrow{\$} \mathbb{Z}_p^*$, $\widetilde{X} \leftarrow \widetilde{G} ^{x}$, $\widetilde{Y} \leftarrow \widetilde{G} ^{y}$, return $(\pk, \sk) \leftarrow ((\widetilde{G} , \widetilde{X}, \widetilde{Y}), (x, y))$.\\
$\Sign_{\PS}(\sk=(x, y), m):$\\
~~~$A \xleftarrow{\$}  \G^*$, $B \leftarrow A^{x + m \cdot y}$, return $\sigma \leftarrow (A, B)$.\\
$\Verify_{\PS}(\pk=(\widetilde{G} , \widetilde{X}, \widetilde{Y}), m, \sigma=(A, B)):$\\
~~~If $A \neq 1_{\G} \land e(A,  \widetilde{X}\widetilde{Y}^m) =e(B, \widetilde{G})$, return $1$. Otherwise return $0$.\\
\hline
\end{tabular}
\caption{\small
The single-message PS signature scheme $\DS_{\PS}$. }
\label{PSsigconst}
\end{figure}

\begin{theorem}[\cite{PS16}]\label{PSfromPSassumption}
If the Pointcheval-Sanders (PS) assumption holds, $\DS_{\PS}$ satisfies the $\rmEUFCMA$ security.
\end{theorem}

\subsection{High-Level Idea of Our Construction}
We give a high-level idea of our synchronized-aggregate signature construction from the PS signature scheme $\DS_{\PS}$.
Let $(\pk_i, \sk_i) = ((\widetilde{G}_i, \widetilde{X}_i, \widetilde{Y}_i), \allowbreak (x_i, y_i))$ be a key pair of the signer $i$ in $\DS_{\PS}$.
The signature $\sigma_i$ on a message $m_i$ signed by $\sk_i$ is formed as $\sigma_i = (A_i, B_i = A_i^{x_i + m_i \cdot y_i})$ where $A_i \xleftarrow{\$} \G^*$.

To construct our synchronized-aggregate signature, we apply the ``public-key sharing technique" and the ``randomness re-use technique" \cite{LOSSW06}.
These techniques are used to construct variants of aggregate signatures \cite{LOSSW06,Sch11,LLY13,CK20,TT20}. 
We explain how to apply these techniques to $\DS_{\PS}$.

First, we consider applying the ``public-key sharing technique".
In this technique, one of element in public key of underlying scheme is replaced by the public parameter.
We change $\pk_i$ as $(\widetilde{X}_i, \widetilde{Y}_i)$ and force signers to use same $\widetilde{G}_i$.
That is, we include $\widetilde{G} = \widetilde{G}_i$ into the public parameter of the scheme.

Second, we consider applying the ``randomness re-use technique".
This technique forces all signers to use the same randomness to sign a message.
If all of signer share same $A_i$, a signature $\sigma$ on a message $m_i$ by each signer $i$ is formed as $(A, B_i = A^{x_i + m_i \cdot y_i})$.
Then, we can compress signatures $\{\sigma_i\}_{i \in [\ell]}$ into an aggregate signature $\Sigma = (A, \prod_{i \in [\ell]} B_i = A^{\sum_{i \in [\ell]} {(x_i + m_i \cdot y_i)}})$.

To share the same randomness $A$ to all signers for each time period $t$, we change $A$ to $H_1(t)$ where $H_1: [T] \rightarrow \G^*$ is a hash function.
Hashing the time as group element has been used to construct variants of aggregate signature schemes \cite{LLY13,LEOM15}.
Moreover, to prove the security, we modify $m_i$ to $H_2(t, m_i)$ where $H_2: [T] \times \{0, 1\}^* \rightarrow \mathbb{Z}_p$ is a hash function.

\subsection{Our Synchronized Aggregate Signature Scheme}\label{Secourcon}
We describe our synchronized aggregate signature scheme $\SAS_{\Ours} = (\Setup_{\Ours} , \allowbreak \KGen_{\Ours} ,\allowbreak \Sign_{\Ours} , \allowbreak\Verify_{\Ours} ,\allowbreak \Agg_{\Ours} , \allowbreak \AVer_{\Ours})$.
The construction of our synchronized aggregate signature scheme is described in Fig.\ref{PSASsigconst}.

\begin{figure}[h]
\centering
\begin{tabular}{|l|}
\hline
$\Setup_{\Ours} (1^\lambda, 1^{T}):$\\
~~~$\BGcal= (p, \G, \widetilde{\G}, \G_T, e) \leftarrow \BG(1^\lambda)$, $\widetilde{G} \xleftarrow{\$} \widetilde{\G}^*$.\\
~~~Choose hash functions: $H_1: [T] \rightarrow \G^*$, $H_2: [T] \times \{0, 1\}^* \rightarrow \mathbb{Z}_p$. \\
~~~Return $\pp \leftarrow (\BGcal, \widetilde{G}, H_1, H_2)$.\\
$\KGen_{\Ours} (\pp):$\\
~~~$x, y \xleftarrow{\$} \mathbb{Z}^*_p$, $\widetilde{X} \leftarrow \widetilde{G}^{x}$, $\widetilde{Y} \leftarrow \widetilde{G}^{y}$, return $(\pk, \sk) \leftarrow ((\widetilde{X}, \widetilde{Y}), (x, y))$.\\
$\Sign_{\Ours} (\sk=(x, y), t, m):$\\
~~~$m' \leftarrow H_2(t,m)$, $B \leftarrow H_1(t)^{x + m' \cdot y}$, return $(B, t)$.\\
$\Verify_{\Ours} (\pk=(\widetilde{X}, \widetilde{Y}), m, \sigma):$\\
~~~$m' \leftarrow H_2(t,m)$, parse $\sigma$ as $(B, t)$.\\
~~~If $e(H_1(t), \widetilde{X}\widetilde{Y}^{m'}) = e(B, \widetilde{G})$, return $1$. Otherwise return $0$.\\
$\Agg_{\Ours} ((\pk_i, m_i, \sigma_i)_{i \in [\ell]}):$\\
~~~For $i=1$ to $\ell$, parse $\sigma_i$ as $(B_i, t_i)$.\\
~~~If there exists $i \in \{2,\dots, \ell\}$ such that $t_i \neq t_1$, return $\bot.$\\
~~~If there exists $(i, j) \in [\ell]\times [\ell]$ such that $i \neq j \land \pk_i = \pk_j$, return $\bot$.\\
~~~If there exists $i \in [\ell]$ suth that $\Verify_{\Ours} (\pk_i, m_i, \sigma_i) \neq 0$, return~$\bot$.\\
~~~$B' \leftarrow \prod^{\ell}_{i=1} B_i$, return $\Sigma \leftarrow (B', t)$.\\
$\AVer_{\Ours} ((\pk_i, m_i)_{i \in [\ell]}, \Sigma):$\\
~~~There exists $(i, j) \in [\ell]\times [\ell]$ such that $i \neq j \land \pk_i = \pk_j$, return $0$.\\
~~~For $i= 1$ to $\ell$, $m'_i \leftarrow H_2(t, m_i)$.\\
~~~Parse $\Sigma$ as $(B', t)$.\\
~~~If $e(H_1(t), (\prod^{\ell}_{i=1}\widetilde{X_i}\widetilde{Y_i}^{m'_i})) = e(B', \widetilde{G})$, return $1$. Otherwise return $0$.\\
\hline
\end{tabular}
\caption{\small
Our synchronized aggregate signature scheme $\SAS_{\Ours}$. }
\label{PSASsigconst}
\end{figure}

\paragraph{\bf Correctness.}
We confirm the correctness of our scheme $\SAS_{\Ours}$. 
Let $\pp \leftarrow \Setup_{\Ours} (1^\lambda, 1^{T})$, $t \in [T]$, $(\pk_i, \sk_i) \leftarrow \KGen_{\Ours} (\pp)$ for $i \in [\ell]$ and $\sigma_i \leftarrow \Sign_{\Ours} (\sk_i, t, m_i)$ for $i \in [\ell]$ where $\pk_i$ are all distinct.
First, we check the correctness of a non-aggregated signature.
For each $i \in [\ell]$,   $B_i = H_1(t)^{x_i +  H_2(t,m_i) \cdot y_i}$ holds where $\sigma_i = (B_i, t)$ and $\sk_i = (x_i, y_i)$.
By these fact, $e(H_1(t), \widetilde{X_i}\widetilde{Y_i}^{H_2(t,m_i)}) = e(B_i, \widetilde{G})$ holds where $\pk_i = (\widetilde{X}_i, \widetilde{Y}_i)$.
Thus, we can see that the correctness of a non-aggregated signature $\sigma_i$ holds. 

Next, we check the correctness of an aggregate signature.
Let $\Sigma = (B', t) \leftarrow \Agg_{\Ours} ((\pk_i, m_i, \sigma_i)_{i \in [\ell]})$.
Then, $B' = \prod^{\ell}_{i=1} B_i = \prod^{\ell}_{i=1}(H_1(t)^{x_i +  H_2(t,m_i) \cdot y_i}) = H_1(t)^{\sum^{\ell}_{i=1}(x_i +  H_2(t,m_i) \cdot y_i)}$ holds.
By these fact, $e(H_1(t), \prod^{\ell}_{i=1}(\widetilde{X_i}\widetilde{Y_i}^{H_2(t,m_i)})) =  e(H_1(t), \widetilde{G}^{\sum^{\ell}_{i=1}(x_i + H_2(t,m_i) \cdot y_i)})  = e(B_i, \widetilde{G})$ holds.
Thus, we can see that the correctness of aggregate signature $\Sigma$ holds.

\subsection{Security Analysis}\label{OurSchemeProof}
As explained in Section \ref{TechOverview}, security proof technique by Lee et al. \cite{LLY13} cannot be applicable.
Instead, we prove  the $\rmEUFCMA$ security of our scheme $\SAS_{\Ours}$ from the GSP assumption.

\begin{theorem}\label{MCLSyncASEUFCMA}
Let $H_1, H_2$ be a hash function of $\SAS_{\Ours}$ in Fig.\ref{PSsigconst} and $T$ is a polynomial in $\lambda$.
If the GPS assumption holds and $H_1, H_2$ are modeled as the random oracle, our scheme $\SAS_{\Ours}$ satisfies the $\rmEUFCMA$ security in the certified-key model.
\end{theorem}

\begin{proof}

Let $\A$ be an $\rmEUFCMA$ security game adversary of the $\SAS_{\Ours}$ scheme with $q_{H_2}$ hash queries to $\mathcal{O}^{H_2}$.
We construct an adversary $\B$ for the GPS security game of  $\BG_{\GPS}$ by using $\A$. 
The construction of $\B$ is given in Fig.\ref{Ourreduction}.

\begin{figure}[htbp]
\centering
\begin{tabular}{|l|}
\hline
$\B^{\mathcal{O}^{\GPS}_{0} (), \mathcal{O}^{\GPS}_{1} (\cdot, \cdot)}(\BGcal, G, \widetilde{G}, \widetilde{X}^*, \widetilde{Y}^*)$\\
~~~$\mathbb{T}_1 \leftarrow \{\}$, $\mathbb{T}_2 \leftarrow \{\}$, $Q \leftarrow \{\}$, $C \leftarrow \{\}$, $L \leftarrow \{\}$, $K \leftarrow \{\}$,\\
~~~$\pp^* \leftarrow (\BGcal, \widetilde{G})$, $\pk^* \leftarrow (\widetilde{X}^*, \widetilde{Y}^*)$, $t \leftarrow 1$\\
~~~$((\pk^*_i, m^*_i)_{i \in [\ell^*]}, \Sigma^*) \leftarrow \A^{\mathcal{O}^{\Cert}(\cdot, \cdot), \mathcal{O}^{H_1}(\cdot), \mathcal{O}^{H_2}(\cdot, \cdot), \mathcal{O}^{\Sign}(\cdot, \cdot)}(\pp^*, \pk^*)$\\
~~~If $\AVer_{\Ours} ((\pk^*_i, m^*_i)_{i \in [\ell^*]}, \Sigma^*) \neq 1$, then abort.\\
~~~If there exists $j \in [\ell^*]$ such that $\pk^*_j \neq \pk^* \land \pk^*_j \notin L$, then abort.\\
~~~If there is no $j^* \in [\ell^*]$ such that $\pk^*_{j^*} = \pk^* \land m^*_{j^*} \notin Q$, then abort.\\
~~~Set $j^* \in [\ell^*]$ such that $\pk^*_{j^*} = \pk^* \land m^*_{j^*} \notin Q$, $\Sigma^* \leftarrow (B^*{}', t^*)$.\\
~~~$m^*_{j^*}{}' \leftarrow H_2(t^*,m^*_{j^*})$\\
~~~\fbox{If $m^*_{j^*}{}' \in C$, then abort.}\\
~~~Retrive $(x_i, y_i) \leftarrow \sk^*_i$ of $\pk^*_i$ from $K$ for $i \in [\ell^*]\backslash \{j^*\}$.\\
~~~$A' \leftarrow H_1(t^*)$, $m'_{i} \leftarrow H_2(t^*,m^*_{i})$ for $i \in [\ell^*]\backslash \{j^*\}$, \\
~~~$B' \leftarrow B^*{}' \cdot \left(A'{}^{\sum_{i \in [\ell^*]\backslash \{j^*\}} (x_i + m'_{i} \cdot y_{i})} \right)^{-1}$.\\
~~~Return $(m^*_{j^*}, A', B')$.\\
\\
$\mathcal{O}^{\Cert}(\pk=(\widetilde{X}, \widetilde{Y}), \sk=(x, y)):$\\
~~~If $(\widetilde{X} = \widetilde{G}^{x}) \land (\widetilde{Y} = \widetilde{G}^{y})$, $L \leftarrow L \cup \{\pk\}$, $K \leftarrow K \cup \{(\pk, \sk)\}$, return $``\accept"$.\\
~~~Otherwise return $``\reject"$.\\
$\mathcal{O}^{H_1}(t_i):$\\
~~~If there is an entry $(t_i, A_i)$ for some $A_i \in  \G^*$ in $\mathbb{T}_1$, return $A_i$.\\
~~~$A_i \leftarrow \mathcal{O}^{\GPS}_{0} ()$, $\mathbb{T}_1 \leftarrow \mathbb{T}_1 \cup \{(t_i, A_i)\}$, return $A_i$.\\
$\mathcal{O}^{H_2}(t_i, m_j):$\\
~~~If there is an entry $(t_i, m_j, m'_{(t_i,j)})$ for some $m'_{(t_i,j)} \in \mathbb{Z}_p$ in $\mathbb{T}_2$, return~$m'_{(t_i,j)}$.\\
~~~$m'_{(t_i,j)} \xleftarrow{\$} \mathbb{Z}_p$, $\mathbb{T}_2 \leftarrow \mathbb{T}_2 \cup \{(t_i, m_j, m'_{(t_i,j)})\}$, return~$m'_{(t_i,j)}$.\\
$\mathcal{O}^{\Sign}(``\Inst", m_j):$\\
~~~$t \notin [T]$, return $\bot$.\\
~~~If $``\Inst" = ``\instskip"$, $t \leftarrow t +1$.\\
~~~If $``\Inst" = ``\instsign"$,\\
~~~~~~If there is no entry $(t, \cdot)$ in $\mathbb{T}_1$, run $\mathcal{O}^{H_1}(t)$.\\
~~~~~~If there is no entry $(t, m_j, \cdot)$ in $\mathbb{T}_2$, run $\mathcal{O}^{H_2}(t, m_j)$.\\
~~~~~~Retrieve entries $(t, A)$ and  $(t, m_j, m'_{(t,j)})$ from $\mathbb{T}_1$ and $\mathbb{T}_2$, respectively.\\ 
~~~~~~$B \leftarrow  \mathcal{O}^{\GPS}_{1} (A, m'_{(t,j)})$.\\
~~~~~~If $B = \bot$, abort the simulation.\\
~~~~~~$Q \leftarrow Q \cup \{m_j\}$, $C \leftarrow C \cup \{m'_{(t,j)}\}$, return $\sigma \leftarrow (B, t)$, $t \leftarrow t +1$.\\
\hline
\end{tabular}
\caption{\small
The reduction $\B$. }
\label{Ourreduction}
\end{figure}

We confirm that if $\B$ does not abort, $\B$ simulates the $\rmEUFCMA$ game for $\SAS_{\Ours}$.
Now, we discuss the distribution of $\pp*$, $\pk^*$, output of oracles $\mathcal{O^{\Cert}}$, $\mathcal{O}^{H_1}$, $\mathcal{O}^{H_2}$, and $\mathcal{O}^{\Sign}$
\begin{itemize}
\item {\bf Distribution of $\pp^*$ and $\pk^*$:}
It is clear that $\B$ simulates $\pp$ and $\pk$ in the $\rmEUFCMA$ game for the $\SAS_{\Ours}$.

\item {\bf Output of $\mathcal{O^{\Cert}}$:} 
It is clear that $\B$ simulates $\mathcal{O^{\Cert}}$ in the $\rmEUFCMA$ game for the $\SAS_{\Ours}$ in the certified-key model.

\item {\bf Output of $\mathcal{O}^{H_1}$:} 
In the original game, hash values of $H_1$ are chosen from $\G^*$ uniformly at random.
In the simulation of $\B$, the hash value $H(t_i)$ is set by $A_i \leftarrow \mathcal{O}^{\GPS}_{0}$.
Since $ \mathcal{O}^{\GPS}_{0}$ samples $A_i$ from $\G^*$ uniformly at random, $\B$ perfectly simulates $\mathcal{O}^{H_1}$.

\item {\bf Output of $\mathcal{O}^{H_2}$:} 
It is clear that $\B$ simulates $\mathcal{O}^{H_2}$.

\item {\bf Output of $\mathcal{O^{\Sign}}$:}
In the simulation of $\B$, by the programming of $\mathcal{O}^{H_1}$ and $\mathcal{O}^{H_2}$,  $H_1(t)=A$ and $H_2(t, m_j) = m'_{(t,j)}$ hold.
If $B \neq \bot$,  $\mathcal{O}^{\GPS}_{1} (A, m'_{(t,j)})$ returns $B = A^{x} \cdot A^{m'_{(t,j)} \cdot y} = H_1(t)^{x + H_2(t, m_j)  \cdot y}$.
Thus if $\B$ does not abort, $\B$ simulate $\mathcal{O^{\Sign}}$.
\end{itemize}
From the above discussion, we can see that $\B$ does not abort, $\B$ can simulate the $\rmEUFCMA$ game for $\SAS_{\Ours}$.

Second, we confirm that if $\A$ successfully output a valid forgery $((\pk^*_i, m^*_i)_{i \in [\ell^*]}, \allowbreak \Sigma^*)$ of $\SAS_{\Ours}$, $\B$ can extract a solution for the GPS problem.
Let $((\pk^*_i, m^*_i)_{i \in [\ell^*]}, \allowbreak \Sigma^*)$ be a valid forgery output by $\A$.
Then there exists $j^* \in [\ell^*]$ such that $\pk^*_{j^*} = \pk^*$.
By the verification of $\AVer_{\Ours}$, 
\begin{equation*}
e(H_1(t^*), (\prod^{\ell}_{i=1}\widetilde{X_i}\widetilde{Y_i}^{H_2(t^*, m_i^*)})) = e(B^*{}', \widetilde{G}) 
\end{equation*}
holds.
If $\B$ does not abort in the procedure \fbox{If $m^*_{j^*}{}' \in C$, then abort.} in Fig.\ref{Ourreduction}, $m^*_{j^*}{}' \in C$ has not been queried to $\mathcal{O}^{\GPS}_{1}$.

We can see that $B^*{}' = A'{}^{\sum^{\ell^*}_{i=1} (x^*_i + y^*_i \cdot m_i^*{}')}$ holds where $(x_i, y_i) = \sk^*_i$ is a secret key corresponding to $\pk^*_i$.
In the certified-key model, since $\B$ knows all $\{\sk^*_i\}_{i \in [\ell^*]\backslash \{j^*\}}$, $\B$ can compute the following.

\begin{equation*}
B' =  A'{}^{x_{j^*} + m'_{j^*} \cdot y_{j^*}}  = B^*{}' \cdot \left(A'{}^{\sum_{i \in [\ell^*]\backslash \{j^*\}} (x_i + m'_{i} \cdot y_{i})} \right)^{-1}
\end{equation*}

Therefore, if $\B$ does not abort,  and $\B$ a solution $(m^*_{j^*}{}', A'{}, B' )$ for the GPS problem.

We analyze the probability that $\B$ succeeds in forging a signature of $\PS$.
First, we consider the probability that $\B$ aborts at the simulation of signatures.
$\B$ aborts the simulation of $\mathcal{O^{\Sign}}$ if $\B$ queries same $A$ at least twice for $\mathcal{O}^{\GPS}_{1} (A, m'_{(t,j)})$.
To give an upper bound of this probability, it is sufficient to consider the probability that collision is found in $H_1$.
We can bound the probability that $\B$ fails simulating a signature for each signing query by $q_s/|\G^*| = q_s/(p-1)$ where $q_s$ is the number of queries to $\mathcal{O^{\Sign}}$ from $\A$.
By taking union bound, the probability that $\B$  fails simulating signatures through the $\rmEUFCMA$  game is upper bounded by $q_s^2/(p-1)$

Next, we consider the probability that $\B$ aborts at \fbox{If $m^*_{j^*}{}' \in C$, then abort.}  in Fig.\ref{Ourreduction}.
This probability can be bounded by the probability that a collision is found in $H_2$.
We can bound this probability by $q_{H_2}/|\Z_p| = q_{H_2}/p$ where $q_{H_2}$ is the number of queries to $\mathcal{O}^{H_2}$. 

Finally, we summarize the above discussion.
Let $\Adv^{\EUFCMA}_{\SAS_{\Ours}, \A}$ be the advantage of the $\rmEUFCMA$ game for the $\SAS_{\Ours}$ scheme of $\A$.
The advantage of the GPS game $\B$ is 
\begin{equation*}
\Adv^{\GPS}_{\BG, \A} \geq \Adv^{\EUFCMA}_{\SAS_{\Ours}, \A}  - \frac{q_s^2}{p-1} - \frac{q_{H_2}}{p}.  
\end{equation*}
Therefore, we can conclude Theorem \ref{MCLSyncASEUFCMA}.
\qed
\end{proof}

\section{Conclusion}
In this paper, we construct the PS signature-based synchronized aggregate signature scheme which offers the most efficient aggregate signature verification among existing synchronized aggregate signature schemes.
As for the security proof of our scheme, since the reduction technique by Lee et a., \cite{LLY13} could not be applied in the security proof of our scheme, we prove its security by using the GPS assumption in the ROM as a new approach.

If we apply the public-key sharing technique and the randomness re-use technique to the CL signature scheme on type-$3$ pairing, we will obtain the CL signature-based synchronized aggregate signature scheme on type-$3$ pairing.
However, as with our PS signature-based synchronized aggregate signature scheme, group elements of a public key and group element in a signature belong to different groups $\widetilde{\G}$ and  $\G$ respectively,  the reduction technique by Lee et al, \cite{LLY13} would not be applied.
Fortunately, similar to the GPS assumption, the generalized LRSW (GLRSW) assumption \cite{CCDLNU17} that is a variant of the LRSW assumption \cite{LRSW99} was proposed.
We leave a future task to confirm whether our reduction technique can be applied to the CL signature-based synchronized aggregate signature scheme on type-$3$ pairing and prove its security from  the GLRSW assumption.

\section*{Acknowledgement}
A part of this work was supported by JST CREST JP-MJCR2113, JSPS KAKENHI JP21H04879, and the technology promotion association of Tsuruoka KOSEN.
We also would like to thank anonymous referees for their constructive comments.

\bibliographystyle{abbrvurl}
\bibliography{PSSync}

\begin{thebibliography}{10}

\bibitem{AFGHO10}
M.~Abe, G.~Fuchsbauer, J.~Groth, K.~Haralambiev, and M.~Ohkubo.
\newblock Structure-preserving signatures and commitments to group elements.
\newblock In T.~Rabin, editor, {\em Advances in Cryptology - {CRYPTO} 2010,
  30th Annual Cryptology Conference, Santa Barbara, CA, USA, August 15-19,
  2010. Proceedings}, volume 6223 of {\em Lecture Notes in Computer Science},
  pages 209--236. Springer, 2010.
\newblock \href {https://doi.org/10.1007/978-3-642-14623-7\_12}
  {\path{doi:10.1007/978-3-642-14623-7\_12}}.

\bibitem{AGH10}
J.~H. Ahn, M.~Green, and S.~Hohenberger.
\newblock Synchronized aggregate signatures: new definitions, constructions and
  applications.
\newblock In {\em Proceedings of the 17th {ACM} Conference on Computer and
  Communications Security, {CCS} 2010, Chicago, Illinois, USA, October 4-8,
  2010}, pages 473--484, 2010.
\newblock \href {https://doi.org/10.1145/1866307.1866360}
  {\path{doi:10.1145/1866307.1866360}}.

\bibitem{ADEO21}
D.~F. Aranha, A.~P.~K. Dalskov, D.~Escudero, and C.~Orlandi.
\newblock Improved threshold signatures, proactive secret sharing, and input
  certification from {LSS} isomorphisms.
\newblock In P.~Longa and C.~R{\`{a}}fols, editors, {\em Progress in Cryptology
  - {LATINCRYPT} 2021 - 7th International Conference on Cryptology and
  Information Security in Latin America, Bogot{\'{a}}, Colombia, October 6-8,
  2021, Proceedings}, volume 12912 of {\em Lecture Notes in Computer Science},
  pages 382--404. Springer, 2021.
\newblock \href {https://doi.org/10.1007/978-3-030-88238-9\_19}
  {\path{doi:10.1007/978-3-030-88238-9\_19}}.

\bibitem{BGOY07}
A.~Boldyreva, C.~Gentry, A.~O'Neill, and D.~H. Yum.
\newblock Ordered multisignatures and identity-based sequential aggregate
  signatures, with applications to secure routing.
\newblock In {\em Proceedings of the 2007 {ACM} Conference on Computer and
  Communications Security, {CCS} 2007, Alexandria, Virginia, USA, October
  28-31, 2007}, pages 276--285, 2007.
\newblock \href {https://doi.org/10.1145/1315245.1315280}
  {\path{doi:10.1145/1315245.1315280}}.

\bibitem{BDN18}
D.~Boneh, M.~Drijvers, and G.~Neven.
\newblock Compact multi-signatures for smaller blockchains.
\newblock In T.~Peyrin and S.~D. Galbraith, editors, {\em Advances in
  Cryptology - {ASIACRYPT} 2018 - 24th International Conference on the Theory
  and Application of Cryptology and Information Security, Brisbane, QLD,
  Australia, December 2-6, 2018, Proceedings, Part {II}}, volume 11273 of {\em
  Lecture Notes in Computer Science}, pages 435--464. Springer, 2018.
\newblock \href {https://doi.org/10.1007/978-3-030-03329-3\_15}
  {\path{doi:10.1007/978-3-030-03329-3\_15}}.

\bibitem{BGLS03}
D.~Boneh, C.~Gentry, B.~Lynn, and H.~Shacham.
\newblock Aggregate and verifiably encrypted signatures from bilinear maps.
\newblock In {\em Advances in Cryptology - {EUROCRYPT} 2003, International
  Conference on the Theory and Applications of Cryptographic Techniques,
  Warsaw, Poland, May 4-8, 2003, Proceedings}, pages 416--432, 2003.
\newblock \href {https://doi.org/10.1007/3-540-39200-9\_26}
  {\path{doi:10.1007/3-540-39200-9\_26}}.

\bibitem{CCDLNU17}
J.~Camenisch, L.~Chen, M.~Drijvers, A.~Lehmann, D.~Novick, and R.~Urian.
\newblock One {TPM} to bind them all: Fixing {TPM} 2.0 for provably secure
  anonymous attestation.
\newblock In {\em 2017 {IEEE} Symposium on Security and Privacy, {SP} 2017, San
  Jose, CA, USA, May 22-26, 2017}, pages 901--920. {IEEE} Computer Society,
  2017.
\newblock \href {https://doi.org/10.1109/SP.2017.22}
  {\path{doi:10.1109/SP.2017.22}}.

\bibitem{CDLNT20}
J.~Camenisch, M.~Drijvers, A.~Lehmann, G.~Neven, and P.~Towa.
\newblock Short threshold dynamic group signatures.
\newblock In C.~Galdi and V.~Kolesnikov, editors, {\em Security and
  Cryptography for Networks - 12th International Conference, {SCN} 2020,
  Amalfi, Italy, September 14-16, 2020, Proceedings}, volume 12238 of {\em
  Lecture Notes in Computer Science}, pages 401--423. Springer, 2020.
\newblock \href {https://doi.org/10.1007/978-3-030-57990-6\_20}
  {\path{doi:10.1007/978-3-030-57990-6\_20}}.

\bibitem{CL04}
J.~Camenisch and A.~Lysyanskaya.
\newblock Signature schemes and anonymous credentials from bilinear maps.
\newblock In {\em Advances in Cryptology - {CRYPTO} 2004, 24th Annual
  International CryptologyConference, Santa Barbara, California, USA, August
  15-19, 2004, Proceedings}, pages 56--72, 2004.
\newblock \href {https://doi.org/10.1007/978-3-540-28628-8\_4}
  {\path{doi:10.1007/978-3-540-28628-8\_4}}.

\bibitem{CGKN21}
K.~Chalkias, F.~Garillot, Y.~Kondi, and V.~Nikolaenko.
\newblock Non-interactive half-aggregation of eddsa and variants of schnorr
  signatures.
\newblock In K.~G. Paterson, editor, {\em Topics in Cryptology - {CT-RSA} 2021
  - Cryptographers' Track at the {RSA} Conference 2021, Virtual Event, May
  17-20, 2021, Proceedings}, volume 12704 of {\em Lecture Notes in Computer
  Science}, pages 577--608. Springer, 2021.
\newblock \href {https://doi.org/10.1007/978-3-030-75539-3\_24}
  {\path{doi:10.1007/978-3-030-75539-3\_24}}.

\bibitem{CK20}
S.~Chatterjee and R.~Kabaleeshwaran.
\newblock From rerandomizability to sequential aggregation: Efficient signature
  schemes based on {SXDH} assumption.
\newblock In J.~K. Liu and H.~Cui, editors, {\em Information Security and
  Privacy - 25th Australasian Conference, {ACISP} 2020, Perth, WA, Australia,
  November 30 - December 2, 2020, Proceedings}, volume 12248 of {\em Lecture
  Notes in Computer Science}, pages 183--203. Springer, 2020.
\newblock \href {https://doi.org/10.1007/978-3-030-55304-3\_10}
  {\path{doi:10.1007/978-3-030-55304-3\_10}}.

\bibitem{CRSST21}
V.~Cini, S.~Ramacher, D.~Slamanig, C.~Striecks, and E.~Tairi.
\newblock Updatable signatures and message authentication codes.
\newblock In J.~A. Garay, editor, {\em Public-Key Cryptography - {PKC} 2021 -
  24th {IACR} International Conference on Practice and Theory of Public Key
  Cryptography, Virtual Event, May 10-13, 2021, Proceedings, Part {I}}, volume
  12710 of {\em Lecture Notes in Computer Science}, pages 691--723. Springer,
  2021.
\newblock \href {https://doi.org/10.1007/978-3-030-75245-3\_25}
  {\path{doi:10.1007/978-3-030-75245-3\_25}}.

\bibitem{CS20}
R.~Clarisse and O.~Sanders.
\newblock Group signature without random oracles from randomizable signatures.
\newblock In K.~Nguyen, W.~Wu, K.~Lam, and H.~Wang, editors, {\em Provable and
  Practical Security - 14th International Conference, ProvSec 2020, Singapore,
  November 29 - December 1, 2020, Proceedings}, volume 12505 of {\em Lecture
  Notes in Computer Science}, pages 3--23. Springer, 2020.
\newblock \href {https://doi.org/10.1007/978-3-030-62576-4\_1}
  {\path{doi:10.1007/978-3-030-62576-4\_1}}.

\bibitem{FKL18}
G.~Fuchsbauer, E.~Kiltz, and J.~Loss.
\newblock The algebraic group model and its applications.
\newblock In H.~Shacham and A.~Boldyreva, editors, {\em Advances in Cryptology
  - {CRYPTO} 2018 - 38th Annual International Cryptology Conference, Santa
  Barbara, CA, USA, August 19-23, 2018, Proceedings, Part {II}}, volume 10992
  of {\em Lecture Notes in Computer Science}, pages 33--62. Springer, 2018.
\newblock \href {https://doi.org/10.1007/978-3-319-96881-0\_2}
  {\path{doi:10.1007/978-3-319-96881-0\_2}}.

\bibitem{GPS08}
S.~D. Galbraith, K.~G. Paterson, and N.~P. Smart.
\newblock Pairings for cryptographers.
\newblock {\em Discrete Applied Mathematics}, 156(16):3113--3121, 2008.
\newblock \href {https://doi.org/10.1016/j.dam.2007.12.010}
  {\path{doi:10.1016/j.dam.2007.12.010}}.

\bibitem{GR06}
C.~Gentry and Z.~Ramzan.
\newblock Identity-based aggregate signatures.
\newblock In {\em Public Key Cryptography - {PKC} 2006, 9th International
  Conference on Theory and Practice of Public-Key Cryptography, New York, NY,
  USA, April 24-26, 2006, Proceedings}, pages 257--273, 2006.
\newblock \href {https://doi.org/10.1007/11745853\_17}
  {\path{doi:10.1007/11745853\_17}}.

\bibitem{Gha16}
E.~Ghadafi.
\newblock Short structure-preserving signatures.
\newblock In K.~Sako, editor, {\em Topics in Cryptology - {CT-RSA} 2016 - The
  Cryptographers' Track at the {RSA} Conference 2016, San Francisco, CA, USA,
  February 29 - March 4, 2016, Proceedings}, volume 9610 of {\em Lecture Notes
  in Computer Science}, pages 305--321. Springer, 2016.
\newblock \href {https://doi.org/10.1007/978-3-319-29485-8\_18}
  {\path{doi:10.1007/978-3-319-29485-8\_18}}.

\bibitem{Gha21}
E.~Ghadafi.
\newblock Partially structure-preserving signatures: Lower bounds,
  constructions and more.
\newblock In K.~Sako and N.~O. Tippenhauer, editors, {\em Applied Cryptography
  and Network Security - 19th International Conference, {ACNS} 2021, Kamakura,
  Japan, June 21-24, 2021, Proceedings, Part {I}}, volume 12726 of {\em Lecture
  Notes in Computer Science}, pages 284--312. Springer, 2021.
\newblock \href {https://doi.org/10.1007/978-3-030-78372-3\_11}
  {\path{doi:10.1007/978-3-030-78372-3\_11}}.

\bibitem{GV22}
R.~Goyal and V.~Vaikuntanathan.
\newblock Locally verifiable signature and key aggregation.
\newblock In Y.~Dodis and T.~Shrimpton, editors, {\em Advances in Cryptology -
  {CRYPTO} 2022 - 42nd Annual International Cryptology Conference, {CRYPTO}
  2022, Santa Barbara, CA, USA, August 15-18, 2022, Proceedings, Part {II}},
  volume 13508 of {\em Lecture Notes in Computer Science}, pages 761--791.
  Springer, 2022.
\newblock \href {https://doi.org/10.1007/978-3-031-15979-4\_26}
  {\path{doi:10.1007/978-3-031-15979-4\_26}}.

\bibitem{HKKKR16}
G.~Hartung, B.~Kaidel, A.~Koch, J.~Koch, and A.~Rupp.
\newblock Fault-tolerant aggregate signatures.
\newblock In {\em Public-Key Cryptography - {PKC} 2016 - 19th {IACR}
  International Conference on Practice and Theory in Public-Key Cryptography,
  Taipei, Taiwan, March 6-9, 2016, Proceedings, Part {I}}, pages 331--356,
  2016.
\newblock \href {https://doi.org/10.1007/978-3-662-49384-7\_13}
  {\path{doi:10.1007/978-3-662-49384-7\_13}}.

\bibitem{HKW15}
S.~Hohenberger, V.~Koppula, and B.~Waters.
\newblock Universal signature aggregators.
\newblock In {\em Advances in Cryptology - {EUROCRYPT} 2015 - 34th Annual
  International Conference on the Theory and Applications of Cryptographic
  Techniques, Sofia, Bulgaria, April 26-30, 2015, Proceedings, Part {II}},
  pages 3--34, 2015.
\newblock \href {https://doi.org/10.1007/978-3-662-46803-6\_1}
  {\path{doi:10.1007/978-3-662-46803-6\_1}}.

\bibitem{HSW13}
S.~Hohenberger, A.~Sahai, and B.~Waters.
\newblock Full domain hash from (leveled) multilinear maps and identity-based
  aggregate signatures.
\newblock In {\em Advances in Cryptology - {CRYPTO} 2013 - 33rd Annual
  Cryptology Conference, Santa Barbara, CA, USA, August 18-22, 2013.
  Proceedings, Part {I}}, pages 494--512, 2013.
\newblock \href {https://doi.org/10.1007/978-3-642-40041-4\_27}
  {\path{doi:10.1007/978-3-642-40041-4\_27}}.

\bibitem{HW18}
S.~Hohenberger and B.~Waters.
\newblock Synchronized aggregate signatures from the {RSA} assumption.
\newblock In J.~B. Nielsen and V.~Rijmen, editors, {\em Advances in Cryptology
  - {EUROCRYPT} 2018 - 37th Annual International Conference on the Theory and
  Applications of Cryptographic Techniques, Tel Aviv, Israel, April 29 - May 3,
  2018 Proceedings, Part {II}}, volume 10821 of {\em Lecture Notes in Computer
  Science}, pages 197--229. Springer, 2018.
\newblock \href {https://doi.org/10.1007/978-3-319-78375-8\_7}
  {\path{doi:10.1007/978-3-319-78375-8\_7}}.

\bibitem{KLAP21}
H.~Kim, Y.~Lee, M.~Abdalla, and J.~H. Park.
\newblock Practical dynamic group signature with efficient concurrent joins and
  batch verifications.
\newblock {\em J. Inf. Secur. Appl.}, 63:103003, 2021.
\newblock \href {https://doi.org/10.1016/j.jisa.2021.103003}
  {\path{doi:10.1016/j.jisa.2021.103003}}.

\bibitem{KSAP21}
H.~Kim, O.~Sanders, M.~Abdalla, and J.~H. Park.
\newblock Practical dynamic group signatures without knowledge extractors.
\newblock Cryptology ePrint Archive, Paper 2021/351, 2021.
\newblock \url{https://eprint.iacr.org/2021/351}.
\newblock URL: \url{https://eprint.iacr.org/2021/351}.

\bibitem{LLY13}
K.~Lee, D.~H. Lee, and M.~Yung.
\newblock Aggregating cl-signatures revisited: Extended functionality and
  better efficiency.
\newblock In {\em Financial Cryptography and Data Security - 17th International
  Conference, {FC} 2013, Okinawa, Japan, April 1-5, 2013, Revised Selected
  Papers}, pages 171--188, 2013.
\newblock \href {https://doi.org/10.1007/978-3-642-39884-1\_14}
  {\path{doi:10.1007/978-3-642-39884-1\_14}}.

\bibitem{LEOM15}
I.~Leontiadis, K.~Elkhiyaoui, M.~{\"{O}}nen, and R.~Molva.
\newblock {PUDA} - privacy and unforgeability for data aggregation.
\newblock In M.~Reiter and D.~Naccache, editors, {\em Cryptology and Network
  Security - 14th International Conference, {CANS} 2015, Marrakesh, Morocco,
  December 10-12, 2015, Proceedings}, volume 9476 of {\em Lecture Notes in
  Computer Science}, pages 3--18. Springer, 2015.
\newblock \href {https://doi.org/10.1007/978-3-319-26823-1\_1}
  {\path{doi:10.1007/978-3-319-26823-1\_1}}.

\bibitem{LOSSW06}
S.~Lu, R.~Ostrovsky, A.~Sahai, H.~Shacham, and B.~Waters.
\newblock Sequential aggregate signatures and multisignatures without random
  oracles.
\newblock In {\em Advances in Cryptology - {EUROCRYPT} 2006, 25th Annual
  International Conference on the Theory and Applications of Cryptographic
  Techniques, St. Petersburg, Russia, May 28 - June 1, 2006, Proceedings},
  pages 465--485, 2006.
\newblock \href {https://doi.org/10.1007/11761679\_28}
  {\path{doi:10.1007/11761679\_28}}.

\bibitem{LMRS04}
A.~Lysyanskaya, S.~Micali, L.~Reyzin, and H.~Shacham.
\newblock Sequential aggregate signatures from trapdoor permutations.
\newblock In {\em Advances in Cryptology - {EUROCRYPT} 2004, International
  Conference on the Theory and Applications of Cryptographic Techniques,
  Interlaken, Switzerland, May 2-6, 2004, Proceedings}, pages 74--90, 2004.
\newblock \href {https://doi.org/10.1007/978-3-540-24676-3\_5}
  {\path{doi:10.1007/978-3-540-24676-3\_5}}.

\bibitem{LRSW99}
A.~Lysyanskaya, R.~L. Rivest, A.~Sahai, and S.~Wolf.
\newblock Pseudonym systems.
\newblock In {\em Selected Areas in Cryptography, 6th Annual International
  Workshop, SAC'99, Kingston, Ontario, Canada, August 9-10, 1999, Proceedings},
  pages 184--199, 1999.
\newblock \href {https://doi.org/10.1007/3-540-46513-8\_14}
  {\path{doi:10.1007/3-540-46513-8\_14}}.

\bibitem{McD20}
K.~L. McDonald.
\newblock The landscape of pointcheval-sanders signatures: Mapping to
  polynomial-based signatures and beyond.
\newblock Cryptology ePrint Archive, Paper 2020/450, 2020.
\newblock \url{https://eprint.iacr.org/2020/450}.
\newblock URL: \url{https://eprint.iacr.org/2020/450}.

\bibitem{PS16}
D.~Pointcheval and O.~Sanders.
\newblock Short randomizable signatures.
\newblock In {\em Topics in Cryptology - {CT-RSA} 2016 - The Cryptographers'
  Track at the {RSA} Conference 2016, San Francisco, CA, USA, February 29 -
  March 4, 2016, Proceedings}, pages 111--126, 2016.
\newblock \href {https://doi.org/10.1007/978-3-319-29485-8\_7}
  {\path{doi:10.1007/978-3-319-29485-8\_7}}.

\bibitem{PS18}
D.~Pointcheval and O.~Sanders.
\newblock Reassessing security of randomizable signatures.
\newblock In {\em Topics in Cryptology - {CT-RSA} 2018 - The Cryptographers'
  Track at the {RSA} Conference 2018, San Francisco, CA, USA, April 16-20,
  2018, Proceedings}, pages 319--338, 2018.
\newblock \href {https://doi.org/10.1007/978-3-319-76953-0\_17}
  {\path{doi:10.1007/978-3-319-76953-0\_17}}.

\bibitem{San20}
O.~Sanders.
\newblock Efficient redactable signature and application to anonymous
  credentials.
\newblock In A.~Kiayias, M.~Kohlweiss, P.~Wallden, and V.~Zikas, editors, {\em
  Public-Key Cryptography - {PKC} 2020 - 23rd {IACR} International Conference
  on Practice and Theory of Public-Key Cryptography, Edinburgh, UK, May 4-7,
  2020, Proceedings, Part {II}}, volume 12111 of {\em Lecture Notes in Computer
  Science}, pages 628--656. Springer, 2020.
\newblock \href {https://doi.org/10.1007/978-3-030-45388-6\_22}
  {\path{doi:10.1007/978-3-030-45388-6\_22}}.

\bibitem{San21}
O.~Sanders.
\newblock Improving revocation for group signature with redactable signature.
\newblock In J.~A. Garay, editor, {\em Public-Key Cryptography - {PKC} 2021 -
  24th {IACR} International Conference on Practice and Theory of Public Key
  Cryptography, Virtual Event, May 10-13, 2021, Proceedings, Part {I}}, volume
  12710 of {\em Lecture Notes in Computer Science}, pages 301--330. Springer,
  2021.
\newblock \href {https://doi.org/10.1007/978-3-030-75245-3\_12}
  {\path{doi:10.1007/978-3-030-75245-3\_12}}.

\bibitem{ST21}
O.~Sanders and J.~Traor{\'{e}}.
\newblock {EPID} with malicious revocation.
\newblock In K.~G. Paterson, editor, {\em Topics in Cryptology - {CT-RSA} 2021
  - Cryptographers' Track at the {RSA} Conference 2021, Virtual Event, May
  17-20, 2021, Proceedings}, volume 12704 of {\em Lecture Notes in Computer
  Science}, pages 177--200. Springer, 2021.
\newblock \href {https://doi.org/10.1007/978-3-030-75539-3\_8}
  {\path{doi:10.1007/978-3-030-75539-3\_8}}.

\bibitem{Sch11}
D.~Schr{\"{o}}der.
\newblock How to aggregate the {CL} signature scheme.
\newblock In {\em Computer Security - {ESORICS} 2011 - 16th European Symposium
  on Research in Computer Security, Leuven, Belgium, September 12-14, 2011.
  Proceedings}, pages 298--314, 2011.
\newblock \href {https://doi.org/10.1007/978-3-642-23822-2\_17}
  {\path{doi:10.1007/978-3-642-23822-2\_17}}.

\bibitem{SSKP22}
M.~Sedaghat, D.~Slamanig, M.~Kohlweiss, and B.~Preneel.
\newblock Structure-preserving threshold signatures.
\newblock Cryptology ePrint Archive, Paper 2022/839, 2022.
\newblock \url{https://eprint.iacr.org/2022/839}.
\newblock URL: \url{https://eprint.iacr.org/2022/839}.

\bibitem{Sho97}
V.~Shoup.
\newblock Lower bounds for discrete logarithms and related problems.
\newblock In W.~Fumy, editor, {\em Advances in Cryptology - {EUROCRYPT} '97,
  International Conference on the Theory and Application of Cryptographic
  Techniques, Konstanz, Germany, May 11-15, 1997, Proceeding}, volume 1233 of
  {\em Lecture Notes in Computer Science}, pages 256--266. Springer, 1997.
\newblock \href {https://doi.org/10.1007/3-540-69053-0\_18}
  {\path{doi:10.1007/3-540-69053-0\_18}}.

\bibitem{TT20}
M.~Tezuka and K.~Tanaka.
\newblock Improved security proof for the camenisch-lysyanskaya signature-based
  synchronized aggregate signature scheme.
\newblock In J.~K. Liu and H.~Cui, editors, {\em Information Security and
  Privacy - 25th Australasian Conference, {ACISP} 2020, Perth, WA, Australia,
  November 30 - December 2, 2020, Proceedings}, volume 12248 of {\em Lecture
  Notes in Computer Science}, pages 225--243. Springer, 2020.
\newblock \href {https://doi.org/10.1007/978-3-030-55304-3\_12}
  {\path{doi:10.1007/978-3-030-55304-3\_12}}.

\bibitem{TT22}
M.~Tezuka and K.~Tanaka.
\newblock Pointcheval-sanders signature-based synchronized aggregate signature.
\newblock In S.~Seung-Hyun and H.~Seo, editors, {\em Information Security and
  Cryptology, {ICISC} 2022, Seoul, South Korea, November 30 - December 2, 2022,
  Proceedings}, volume 13849 of {\em Lecture Notes in Computer Science}, pages
  317--336. Springer, 2023.
\newblock \href {https://doi.org/10.1007/978-3-031-29371-9\_16}
  {\path{doi:10.1007/978-3-031-29371-9\_16}}.

\end{thebibliography}

\newpage
\setcounter{tocdepth}{2}
\tableofcontents

\end{document}